\documentclass[12pt]{article}
\usepackage{hyperref}
\usepackage{graphicx}
\usepackage[english]{babel}
\usepackage{amsmath,amsfonts,amssymb} 

\newtheorem{theorem}{Theorem}
\newtheorem{lemma}{Lemma}
\newtheorem{proof}{Proof}

\begin{document}
\title{On extrema of the objective functional for short-time generation of single qubit quantum gates}
\author{Alexander Pechen$^{1}$\footnote{Corresponding author. E-mail: \href{apechen@gmail.com}{apechen@gmail.com}; Webpage: \href{http://www.mathnet.ru/eng/person17991}{mathnet.ru/eng/person17991}} \, and Nikolay Il'in$^{2}$}
\date{}
\maketitle

\vspace{-1cm}
\begin{center}
$^{1}$Department of Mathematical Methods for Quantum Technologies,\\
Steklov Mathematical Institute of Russian Academy of Sciences\\ Gubkina str., 8, Moscow 119991, Russia\\
$^{2}$The National University of Science and Technology MISiS\\
Leninsky Prospekt, 4, Moscow 119049, Russia
\end{center}

\begin{abstract}
In the present work the extrema of the objective functional for the problem of  generation of quantum gates (logical elements for quantum computations) for two--level systems are investigated for short duration of the control. The problem of  existence of local but not global extrema, the so called traps, is considered. In prior works the absence of traps was proved for a sufficiently long control duration. In this paper we prove that for almost all target unitary operators and system Hamiltonians traps are absent for an arbitrarily small control duration. For the remainder target unitary operators and Hamiltonians we obtain a new estimate for the lower boundary of the control duration which guarantees the absence of traps.
\end{abstract}

\section*{Introduction}
In the present paper we study the problem of controlling a qubit, i.e. a two-level quantum
system, by using coherent control pulses (electromagnetic
field). Qubit is one of basic elements for realization of quantum
computing and for creation of quantum computer. An important problem is to generate
single qubit gates (logical elements for quantum computation)~\cite{1}.

The dynamics of a qubit interacting with coherent control $f(t)$ under the assumption 
of good enough isolation of the qubit from the environment is described by the Schr\"odinger equation for unitary evolution $U_t$ (a $2\times 2$ unitary matrix):
\begin{equation}
\label{eq1}
i\,\frac{dU_t}{dt}=(H_0+f(t)V)U_t,\qquad U_{t=0}=\mathbb I.
\end{equation}
Here $H_0$ and $V$ are $2\times 2$ Hermitian matrices. In order to make the control problem non-trivial we assume that $[H_0,V]\ne0$. In the presence of the environment the reduced dynamics of the qubit is described by various master equations~\cite{2,3,4} and by a quantum channel instead of a unitary transformation~\cite{5,6}. The control $f$ belongs to some set of admissible controls~$\mathcal U$, $f\in\mathcal U$. In the applications one considers sets of admissible controls $\mathcal U=L^1([0,T];\mathbb R)$, $\mathcal U=L^2([0,T];\mathbb R)$ and others. Here $T>0$ is the fixed control duration. The problem of optimal performance is also considered~\cite{7}. In this work we consider the set of control $\mathcal U=L^1([0,T];\mathbb R)$. Matrix elements $[U_t]_{ik}$ are assumed to be absolutely continuous functions on the interval $[0,T]$, $[U_t]_{ik}\in\mathrm{AC}[0,T]$.  In this case the equation~\eqref{eq1} has a unique solution for every control $f\in\mathcal U$~\cite{8}.

An important problem in quantum information is generation of quantum gates (special unitary $(2\times 2)$--matrices) $W\in\mathrm{SU}(2)$, i.e. the search for such a control $f$ that $U_T=W$, perhaps up to a phase factor.
This problem can be formulated as the problem of finding a control $f$ which maximizes the objective functional
\begin{equation}
\label{eq2}
\mathcal J_W[f] = \frac14|\operatorname{Tr}(W^\dag U_T)|^2.
\end{equation}
The objective functional $\mathcal J_W$ reaches its maximum value
$\mathcal J_W^\mathrm{max}=1$ on a unitary matrix of the form  $U_T=We^{i\omega}$,
where $\omega\in\mathbb R$ is an arbitrary phase. The global minimum of the objective $\mathcal J_W$ is equal to zero, $\mathcal J_W^\mathrm{min}=0$.
Examples of the objective matrix $W$ which are important for applications include the Hadamard gate
$W=\mathbb H$,
\begin{equation}
\label{eq3}
\mathbb H=\frac1{\sqrt{2}}
\begin{pmatrix}
1&1
\\
1&-1
\end{pmatrix},
\end{equation}
the phase shift gate $W=U_\phi$, where $\phi\in(0,2\pi)$,
\begin{equation}
\label{eq4}
U_\phi=\begin{pmatrix}
1&0\\0&e^{i\phi}
\end{pmatrix}
\end{equation}
and other.

In the present work we consider the problem of possible existence of local but not global maxima for the objective functional $\mathcal J_W$ for short $T$. Such local maxima are called traps~\cite{9}--\cite{13}.
We prove  the absence of traps for almost all objective unitary operators and system Hamiltonians for an arbitrarily small control duration. For the remainder  set of objective unitary operators and  Hamiltonians we obtain the new estimate for the lower boundary of the control duration, which guarantees the absence of traps.

\section{The absence of traps for controlling a qubit at long times}\label{s1}

If traps would exist they would become the obstacle for the search of globally optimal control by local search algorithms. In works~\cite{9,10} the  absence of traps was conjectured  for typical control problems for systems which are isolated from the environment, i.e. for closed quantum systems~\cite{14,15}. In works~\cite{16,17,18} the absence of traps was proved for two--levels closed quantum systems in the case of sufficiently long $T$.

Define the special control $f_0$  and time $T_0$: 
\begin{align}
f_0&:=\frac{-\operatorname{Tr}H_0\operatorname{Tr}V+2\operatorname{Tr}(H_0V)}
{(\operatorname{Tr}V)^2-2\operatorname{Tr}V^2},
\label{eq5}
\\
T_0&:=\frac{\pi}{\|H_0-\mathbb I\operatorname{Tr}H_0/2+f_0(V-\mathbb I\operatorname{Tr}V/2)\|}.
\label{eq6}
\end{align}
Here and below the norm of a matrix $A$ is the operator norm
$$
\|A\|=\sup_{\|\mathbf a\|=1}\|A\mathbf a\|.
$$
Note that $T_0<\infty$, because if $H_0=\mathbb I\operatorname{Tr}H_0/2-f_0V$
then $[H_0,V]=0$, that contradicts the assumption of non-triviality of the system Hamiltonian.

If $\operatorname{Tr} H_0\ne0$ and $\operatorname{Tr} V\ne0$, then replacing
$H_0$ and $V$ by $\widetilde H_0=H_0-\operatorname{Tr} H_0/2$
and $\widetilde V=V-\operatorname{Tr} V/2$ we can  transform the free Hamiltonian 
to the form with $\operatorname{Tr} \widetilde H_0=0$ and $\operatorname{Tr}\widetilde V=0$.
Such a replacement does not affect the existence of traps, because the evolution operator
$U_T$, which is determined by the solution of the equation~\eqref{eq1} for the pair
$(H_0,V)$, is related to the evolution operator $\widetilde U_T$, which is determined by the solution of the equation~\eqref{eq1} for the pair $(\widetilde H_0,\widetilde V)$, by the equality
$\widetilde U_T=U_Te^{-i\lambda(T)\mathbb I}$. Here
$\lambda(T)=\bigl(T\operatorname{Tr} H_0+\operatorname{Tr} V\int_0^Tf(t)\,dt\bigr)/2$.
Hence, $U_T$ differs from $\widetilde U_T$ by a phase and the objective value under such  replacement 
does not change:
$$
|\operatorname{Tr}(W^\dag\widetilde U_T)|^2=\bigl|e^{-i\lambda(T)}
\operatorname{Tr}(W^\dag U_T)\bigr|^2=|\operatorname{Tr}(W^\dag U_T)|^2.
$$
Unless otherwise stated, without loss of generality below we assume that the matrices $H_0$ and $V$ are traceless. We will also use the Pauli matrices $\sigma_x$,
$\sigma_y$ and $\sigma_z$:
\begin{equation}
\label{eq7}
\sigma_x=\begin{pmatrix} 0 & 1 \\ 1 & 0 \end{pmatrix}, \quad
\sigma_y=\begin{pmatrix} 0 & -i \\ i & 0 \end{pmatrix}, \quad
\sigma_z=\begin{pmatrix} 1 & 0 \\ 0 & -1 \end{pmatrix}.
\end{equation}

In~\cite{17} the following statement is proved.

\begin{theorem}
\label{t1}
If $[H_0,V]\ne0$ and $T\ge T_0$, then all maxima of the objective functional
$\mathcal J_W$ are global. Any control $f\ne f_0$ is not a trap  for any $T>0$.
\end{theorem}

From the results of the paper~\cite{17} it follows that for small $T$ only the special control
$f=f_0$ may be a trap. Below in this work we will show that for almost all
$W$ the control $f_0$ is not a trap for any $T>0$.

\section{The absence of traps for small $T$}\label{s2}

According to Theorem~\ref{t1}, the only potential trap can be the
control $f=f_0$. Therefore, to explore the possibility of the existence of traps
for small $T$ it is sufficient to investigate the behaviour of the objective  at this point.

Theorem~\ref{t2} states the absence of traps for the objective $\mathcal J_W$ for almost all $W$. Let
$d=\|H_0+f_0V-\mathbb I\operatorname{Tr} H_0/2-f_0\mathbb I\operatorname{Tr} V/2\|$.
Note, that any matrix $W$ such that $[H_0+f_0V,W]=0$ has the form 
\begin{equation}
\label{eq8}
W=e^{i\alpha_W (H_0+f_0V)+i\beta_W}, \qquad \alpha_W\in\biggl(0,\frac{\pi}{d}\biggr],\quad
\beta_W\in[0,2\pi).
\end{equation}

\begin{theorem}
\label{t2}
Let in the equation~\eqref{eq1} be $[H_0,V]\ne 0$. Let
$[H_0+\nobreak f_0V,W]\ne0$.
Then for any $T>0$ all maxima of the objective functional $\mathcal J_W$ are global.
Let $[H_0+f_0V,W]=0$. If $\alpha_W\in(0,\pi/(2d))$, then all maxima
of the objective functional $\mathcal J_W$ are global for any $T>0$. If
$\alpha_W\in[\pi/(2d),\pi/d]$, then all maxima
of the objective functional $\mathcal J_W$
are global for any $T>\pi/d-\alpha_W$.
\end{theorem}

Proof of the theorem~\ref{t2} is based on the lemmas~\ref{l1}--\ref{l4}.

We will use the expansion of the objective functional $\mathcal J_W$ in Taylor series
up to the second-order term~\cite{19}.

\begin{lemma}
\label{l1}
There is an asymptotic expansion
\begin{align}
\mathcal J_W[f+\delta f]&=\mathcal J_W[f]+\int_0^T\frac{\delta\mathcal J_W}{\delta
f(t)}\delta f(t)\,dt
\nonumber
\\
&\qquad+\frac12\int_0^T\int_0^T\frac{\delta^2\mathcal J_W}{\delta f(t_2)\delta f(t_1)}
\delta f(t_1)\delta f(t_2)\,dt_1\,dt_2+o(\|\delta f\|^2_{L^1}).
\label{eq9}
\end{align}
Here
\begin{gather}
\frac{\delta\mathcal J_W}{\delta f(t)}=\frac12\operatorname{Im}\bigl(\operatorname{Tr}
Y^\dag\operatorname{Tr}(YV_t)\bigr),\qquad Y=W^\dag U_T,\quad V_t=U_t^\dag V U_t,
\label{eq10}
\\
\frac{\delta^2\mathcal J_W}{\delta f(t_2)\delta f(t_1)}=\begin{cases}\dfrac12
\operatorname{Re}\bigl(\operatorname{Tr}(Y V_{t_1})\operatorname{Tr}(Y^\dag V_{t_2})
-\operatorname{Tr}(Y V_{t_2}V_{t_1})\operatorname{Tr} Y^\dag\bigr), &t_2\ge t_1,
\\[2mm]
\dfrac12\operatorname{Re}\bigl(\operatorname{Tr}(Y V_{t_2})\operatorname{Tr}(Y^\dag V_{t_1})
-\operatorname{Tr}(Y V_{t_1}V_{t_2})\operatorname{Tr} Y^\dag\bigr), &t_2<t_1.\end{cases}
\label{eq11}
\end{gather}
The linear map $A\colon L^1([0,T];\mathbb R)\mapsto\mathbb R$ which is defined as
\begin{equation}
\label{eq12}
Ag=\int_0^T\frac{\delta\mathcal J_W}{\delta f(t)}g(t)\,dt,
\end{equation}
is the Frechet differential of the map $f\to\mathcal J_W[f]$.
\end{lemma}

\begin{proof}
The evolution operator $U_t^f$ induced by the control $f$
satisfies the Schr\"odinger equation
\begin{equation}
\label{eq13}
i\,\frac{dU_t^f}{dt}=(H_0+fV)U_t^f,
\end{equation}
The evolution operator $U_t^{f+g}$ induced by the control $f+g$ satisfies the  equation
\begin{equation}
\label{eq14}
i\,\frac{dU_t^{f+g}}{dt}=(H_0+fV)U_t^{f+g}+gVU_t^{f+g}.
\end{equation}
Making the replacement $U_t^{f+g}=U_t^{f}Z_t$ in the equation~\eqref{eq14}, we obtain
\begin{equation}
\label{eq15}
i\,\frac{dZ_t}{dt}=gV_tZ_t,\qquad V_t=U_t^{f\dag}VU^f_t.
\end{equation}
Now let represent~\eqref{eq15} in the integral form
\begin{equation}
\label{eq16}
Z_t=\mathbb I-i\int_0^{t}V_{t_1}Z_{t_1}g(t_1)\,dt_1.
\end{equation}
Iterating the expression~\eqref{eq16} and multiplying by  $U^f_t$ on the left, we obtain
\begin{align}
U_T^{f+g}&=U_T^f -i\int_0^TU_T^fV_t
g(t)\,dt-\int_0^T\int_0^{t_1}U_T^fV_{t_1}V_{t_2}g(t_1)g(t_2)\,dt_2\,dt_1
\nonumber
\\
&\qquad+i\int_0^T\int_0^{t_1}\int_0^{t_2}U_T^fV_{t_1}V_{t_2}Z_{t_3} g(t_1)g(t_2)g(t_3)\,dt_3\,dt_2\,dt_1.
\label{eq17}
\end{align}
Because $\|U_T^f\|=\|Z_t\|= 1$ and $\|V_{t}\|=\|V\|$, for the last summand in~\eqref{eq17} we obtain the estimate 
\begin{equation}
\label{eq18}
\biggl\|\int_0^T\int_0^{t_1}\int_0^{t_2}U_T^fV_{t_1}V_{t_2}Z_{t_3}
g(t_1)g(t_2)g(t_3)\,dt_3\,dt_2\,dt_1\biggr\|\le\|g\|^3_{L^1([0,T];\mathbb R)}\|V\|^2.
\end{equation}
For the first and second order variations we have the estimates  
\begin{gather}
\biggl\|\int_0^TU_T^fV_t g(t_1)\,dt_1\biggr\|\le \|g\|_{L^1([0,T];\mathbb R)}\|V\|,
\label{eq19}
\\
\biggl\|\int_0^T\int_0^{t_1}U_T^fV_{t_1}V_{t_2}g(t_1)g(t_2)\,dt_2\,dt_1
\biggr\|\le \|g\|^2_{L^1([0,T];\mathbb R)}\|V\|^2.
\label{eq20}
\end{gather}
Replacing in the objective  
\begin{equation}
\mathcal J_W[f+\delta f]=\frac14\operatorname{Tr}\bigl(W^\dag U_T^{f+\delta f}\bigr)
\operatorname{Tr}\bigl(WU_T^{\dag f+\delta f}\bigr)
\end{equation}
the expression~\eqref{eq17}, we obtain the asymptotic expansion~\eqref{eq9}.
From~\eqref{eq9}, because $\delta\mathcal J_W/\delta f(t)$
and~${\delta^2\mathcal J_W}/{\delta f(t_2)\delta f(t_1)}$ are bounded functions, we obtain
\begin{gather*}
\mathcal J_W[f+g]=\mathcal J_W[f]+Ag+o(\|g\|_{L^1}),
\\
Ag=\int_0^T\frac{\delta\mathcal J_W}{\delta f(t)}g(t)\,dt.
\end{gather*}
Therefore bounded linear operator $A\colon L^1([0,T];\mathbb R)\mapsto\mathbb R$
is the Frechet differential of the map $f\to\mathcal J_W[f]$.
This proves the lemma.
\end{proof}

Necessary conditions for the point $f=f_0$ to be a maximum or a minimum of the objective  functional $\mathcal J_W$
are determined by vanishing of the gradient
$\delta \mathcal J_W/\delta f\big|_{f=f_0}=0$ and by semi-definiteness of the quadratic form
$$
\int_0^T\int_0^T\operatorname{Hess}(\tau_2,\tau_1)f(\tau_2)f(\tau_1)\,d\tau_1\,d\tau_2\ge0,
\quad\text{where}\ \ \operatorname{Hess}(t_1,t_2)=\frac{\delta^2 \mathcal J_W}
{\delta f(t_2)\delta f(t_1)}\bigg|_{f=f_0}.
$$
A sufficient condition for the control  $f=f_0$ to be a saddle point is determined by vanishing of the gradient $\delta\mathcal J_W/\delta f\big|_{f=f_0}=0$
and by alternating quadratic form
$$
\int_0^T\int_0^T\operatorname{Hess}(\tau_2,\tau_1)f(\tau_2)f(\tau_1)\,d\tau_1\,d\tau_2.
$$

For the analysis of properties of the control problem in a neighbourhood of the special control $f_0$
it is convenient to choose a special basis in the space of $(2\times2)$-matrices, for which
the equation~\eqref{eq1}  has a simple form.
Recall that, without loss of generality, we can set
$\operatorname{Tr}H_0=\operatorname{Tr}V=0$.
Therefore there exists a unitary matrix $S$ such that $S(H_0+f_0V)S^\dag=h\sigma_z$,
where $h=\|H_0+f_0V\|$ and~$SVS^\dag=\mathrm v_x\sigma_x+\mathrm v_y\sigma_y+\mathrm
v_z\sigma_z$.
Note that $\mathrm v_z=0$, because
$\mathrm v_z=\operatorname{Tr}(SVS^\dag\sigma_z)/2=\operatorname{Tr}[V(H_0+f_0V)]/2h=0$,
since
$f_0=-\operatorname{Tr}(H_0V)/\operatorname{Tr}V^2$ under condition
$\operatorname{Tr}H_0=\operatorname{Tr}V=0$.
If we make the replacement $U\to SUS^\dag$ and the replacement of time $t\to t/h$, then
the equation~\eqref{eq1} takes the form
\begin{equation}
\label{eq22}
i\,\frac{dU_t}{dt}=\bigl(\sigma_z+g(t)(v_x\sigma_x+v_y\sigma_y)\bigr)U_t.
\end{equation}
Here $g(t)=f(t)-f_0$, $v_x=\mathrm v_x/h$, $v_y=\mathrm v_y/h$. For the system~\eqref{eq22} we have $f_0=0$ and $T_0=\pi$.
In the transition to a new basis the matrix $W$ undergoes transformation $W\to SWS^\dag$
which does not change the objective  values, so that
$|\operatorname{Tr}(W^\dag U_T)|^2/4=|\operatorname{Tr}(SW^\dag U_TS^\dag)|^2/4$.
In particular, critical points of the control problem for the system~\eqref{eq1} correspond to critical points  for~\eqref{eq22}. Therefore in lemmas~\ref{l2}--\ref{l4} we can consider maximization of the objective functional $|\operatorname{Tr}(W^\dag U_T)|^2/4$ for the system~\eqref{eq22}.

\begin{lemma}
\label{l2}
Let $[W,\sigma_z]\ne0$. Then for any $T>0$ the control $g=0$ is not a trap for maximization of the objective  functional  $\mathcal J_W$ for the system~\eqref{eq22}, and all maxima of the objective  functional are global.
\end{lemma}

\begin{proof}
The solution of the equation~\eqref{eq22} with constant control $g=\nobreak0$
has the form $U_t=e^{-i\sigma_z t}$. The matrix $V_t=e^{i\sigma_z
t}Ve^{-i\sigma_z t}$ takes the form
\begin{equation}
\label{eq23}
V_t=v\cos(2t-\phi)\sigma_x-v\sin(2t-\phi)\sigma_y.
\end{equation}
Here $v=\sqrt{v_x^2+v_y^2}$ and $\phi=\arctan(v_y/v_x)$.

The matrix $Y=W^\dag U_T$ is unitary. We can parametrize it by the angles  $\varphi$, $\psi$, $\theta$ and by phase $\omega$ (see~\cite{19}) as
\begin{equation}
\label{eq24}
Y=e^{i\omega}
\begin{pmatrix}
e^{i\varphi}\cos\theta & e^{i\psi}\sin\theta
\\
-e^{-i\psi}\sin\theta & e^{-i\varphi}\cos\theta
\end{pmatrix}.
\end{equation}
The angles $\psi$, $\varphi$ and $\theta$ are expressed in terms of the Euler angles $\varphi'$,
$\theta'$ and $\psi'$, which belong to the domain $0\le\varphi'<2\pi$,
$0<\theta'<\pi$, $-2\pi<\psi'<2\pi$, in the following way: $\varphi=(\varphi'+\psi')/2$,
$\psi=(\varphi'-\psi'+\pi)/2$ and $\theta=\theta'/2$. The overall phase $\omega$ can be omitted,
since in the expression for the gradient~\eqref{eq10} and Hessian~\eqref{eq11} the matrices $Y$
and $Y^\dag$ appear in pairs, so that the gradient and the Hessian are phase invariant.
The expression for the objective at the point $g=0$ through angles $\varphi$, $\psi$ and
$\theta$ has the form
\begin{equation}
\label{eq25}
\mathcal J_W[0]=\frac14|\operatorname{Tr} Y|^2=\cos^2\varphi\cos^2\theta.
\end{equation}
We introduce the notation
\begin{equation}
\label{eq26}
L(X):=\frac12\operatorname{Im}\bigl(\operatorname{Tr} Y^\dag\operatorname{Tr}(YX)\bigr).
\end{equation}
If $g=0$ is an extrema for~\eqref{eq22}, then at this point the gradient of the objective  functional
must vanish identically:
\begin{equation}
\label{eq27}
\frac{\delta\mathcal J_W}{\delta f(t)}\biggl|_{g=0}=
v\cos(2t-\phi)L(\sigma_x)-v\sin(2t-\phi)L(\sigma_y)=0\quad \forall t\in[0,T].
\end{equation}
The equality~\eqref{eq27} can be satisfied only if $L(\sigma_x)=0$ and
$L(\sigma_y)=0$. This imposes the restrictions on the angles $\varphi$, $\psi$ and $\theta$
at the extremal points:
\begin{align}
L(\sigma_x)&=2\cos\varphi\cos\theta\sin\theta\sin\psi=0
\label{eq28},
\\
L(\sigma_y)&=2\cos\varphi\cos\theta\sin\theta\cos\psi=0.
\label{eq29}
\end{align}
If $\cos\varphi\cos\theta=0$, then according to~\eqref{eq25} the objective
reaches its global minima $\mathcal J_W=0$. If $\cos\varphi\cos\theta\ne0$,
then the satisfaction of the equalities~\eqref{eq28} and~\eqref{eq29} requires that
$\sin\theta=0$. Under this condition $\cos\theta=\pm 1$ and up to a phase multiplier we have
 $Y=e^{i\sigma_z\varphi}$. Then
\begin{equation}
\label{eq30}
W=U_TY^{\dag}=e^{-i\sigma_z T}Y^\dag=e^{-i\sigma_z (T+\varphi)}.
\end{equation}
Thus, in this case $[W,\sigma_z]=0$ and, therefore, the assumption of the
lemma is not satisfied. This proves the lemma.
\end{proof}

Recall that for the system~\eqref{eq22} the special time is $T_0=\pi$ and for $T\ge T_0$ according to theorem~\ref{t1} traps are absent. The lemma~\ref{l2} states that if
$[W,\sigma_z]\ne 0$ then traps are absent for any $T>0$.
Here we show that if $[W,\sigma_z]=0$ then the lower bound for the time $T$, for which there are no trap, can be reduced.

\begin{lemma}
\label{l3}
For $T\ge\pi/2$,  for any $W$ the control $g=0$ is not a trap for maximization of the objective functional $\mathcal J_W$ for the system~\eqref{eq22}.
\end{lemma}

\begin{proof}
Under the condition $\sin\theta=0$ the Hessian has the form
\begin{equation}
\label{eq31}
\operatorname{Hess}(t_2,t_1)=-2v^2\cos\varphi\cos(2|t_2-t_1|+\varphi).
\end{equation}
We introduce the auxiliary function
\begin{equation}
\label{eq32}
\delta_\varepsilon (t)=\begin{cases}
0, & |t|\ge\dfrac{\varepsilon}2,
\\
\dfrac1{\varepsilon}, & |t|<\dfrac{\varepsilon}2.
\end{cases}
\end{equation}
Then for all $f\in\mathrm C[0,T]$ and $t\in(\varepsilon/2,T-\varepsilon/2)$ we have
\begin{equation}
\label{eq33}
\int_0^T\delta_\varepsilon(\tau-t)f(\tau)\,d\tau=f(t)+O(\varepsilon).
\end{equation}
Let
$$
f_\varepsilon(t)=\lambda\delta_\varepsilon(t-t_1)+\mu\delta_\varepsilon (t-t_2),
\qquad\varepsilon/2<t_1<t_2<T-\varepsilon/2,\quad \varepsilon<t_2-t_1.
$$
Substituting the function $f_\varepsilon(t)$ in the expression for the second variation of the objective functional, we obtain
\begin{align}
(f,\operatorname{Hess}f)&=
\int_0^T\int_0^T\operatorname{Hess}(\tau_2,\tau_1)f(\tau_2)f(\tau_1)\,d\tau_1\,d\tau_2
\nonumber
\\
&=-2v^2\cos\varphi G(\lambda,\mu)+O(\varepsilon),
\label{eq34}
\end{align}
where
\begin{equation}
\label{eq35}
G(\lambda,\mu)=\lambda^2\cos\varphi+2\lambda\mu\cos(2|t_2-t_1|+\varphi)+\mu^2\cos\varphi.
\end{equation}
Bilinear form $G(\lambda,\mu)$ is alternating if and only if its discriminant $D$ is positive:
\begin{equation}
\label{eq36}
D=D(|t_2-t_1|)=\cos^2(2|t_2-t_1|+\varphi)-\cos^2\varphi>0.
\end{equation}
If $\cos^2\varphi=0$ or $\cos^2\varphi=1$, then the objective functional has the global extrema $\mathcal J_W=0$ or $\mathcal J_W=1$. Hence,
trap may correspond  only to those angles $\varphi$ for which
$0<\cos^2\varphi<1$. Then $D$ as function of the difference  $|t_2-t_1|$ takes positive values at some points in the interval  $[0,\pi/2]$, and because it has a period $\pi/2$, its maximum values are positive and minimum values are negative.
Let $t_1,t_2\in [0,\pi/2]$ be such that $G(\lambda,\mu)$ is an alternating form. For such $t_1$ and $t_2$ choose $\lambda_1$, $\lambda_2$, $\mu_1$ and
$\mu_2$ so that $G(\lambda_1,\mu_1)>0$ and $G(\lambda_2,\mu_2)<0$. Let
\begin{align*}
f_{1,\varepsilon}(t)&=\lambda_1\delta_\varepsilon (t-t_1)+\mu_1\delta_\varepsilon
(t-t_2),
\\
f_{2,\varepsilon}(t)&=\lambda_2\delta_\varepsilon (t-t_1)+\mu_2\delta_\varepsilon
(t-t_2),
\end{align*}
where $\varepsilon$ is such that signs of $(f_{j,\varepsilon},
\operatorname{Hess}f_{j,\varepsilon})$, $j=1,2$, coincide with signs of
$\lim_{\varepsilon\to0}(f_{j, \varepsilon},\operatorname{Hess}f_{j,\varepsilon})$.
Then $(f_{1,\varepsilon},\operatorname{Hess}f_{1,\varepsilon})$ and
$(f_{2,\varepsilon},\operatorname{Hess}f_{2,\varepsilon})$ will have
opposite signs, i. e. Hessian at the point $f=0$ is not sign definite.
This proves the lemma.
\end{proof}

Consider matrix $W$ of the form
\begin{equation}\label{W}
W=e^{i\sigma_z\varphi_W}.
\end{equation}
Note that we can put $\varphi_W\in(0,\pi]$. In the opposite case
$\varphi_W=\varphi_W'+\pi k$, $\varphi_W'\in[0,\pi]$ and
$$
W=e^{i\sigma_z\varphi_W}=e^{i\sigma_z\varphi_W'}(-1)^k=W'e^{i\pi k},
$$
that differs by a non-significant phase factor. The corresponding matrix
$Y$ has the form
\begin{equation}
\label{eq38}
Y=e^{-i\sigma_z(\varphi_W+T)},
\end{equation}
i.\,e.\ $\varphi=-\varphi_W-T$.

\begin{lemma}
\label{l4}
If $\varphi_W\in(0,\pi/2)$, then for any $T>0$ the control $g=0$ is not a trap for maximization of the objective functional $\mathcal J_W$ for the system~\eqref{eq22}. If $\varphi_W\in[\pi/2,\pi]$, then for any
$T>\pi-\varphi_W$ the control $g=0$ is not a trap for maximization of the objective functional $\mathcal J_W$ for the system~\eqref{eq22}.
\end{lemma}

\begin{proof}
The Hessian of the objective functional $\mathcal J_W[f]$ at the point $g=0$ has the form~\eqref{eq34},
where the bilinear form $G(\lambda,\mu)$ is defined by the equality~\eqref{eq35}. The Hessian is sign definite if the discriminant $D$, which is defined by the formula~\eqref{eq36},
is positive. Let us demonstrate that there exist $0<t_1$ and $t_2<T$ such that $D>0$.

The expression~\eqref{eq36} for $D$  can be rewritten as
\begin{equation}
\label{eq39}
D=\sin2(\varphi_W+T-|t_2-t_1|)\sin2|t_2-t_1|.
\end{equation}
If $\varphi_W\in(0,\pi/2)$, then for any $T$, $\pi/2>T>0$ chose
$|t_2-t_1|=T-\varepsilon$, where $\varepsilon$ is sufficiently small such that
$(\varphi_W+\varepsilon)\in(0,\pi/2)$ and $T-\varepsilon>0$. Then
\begin{equation}
\label{eq40}
D=\sin2(\varphi_W+\varepsilon)\sin2(T-\varepsilon)>0,
\end{equation}
because  $(\varphi_W+\varepsilon)\in(0,\pi/2)$ and $T<\pi/2$.

If $\varphi_W\in[\pi/2,\pi)$, then in order to obtain $D>0$ it is sufficient to satisfy the
inequalities
\begin{equation}
\label{eq41}
\pi<\varphi_W+T-|t_2-t_1|<\frac{3\pi}2, \quad |t_2-t_1|<\frac{\pi}2.
\end{equation}
For any $T$ which satisfies the  inequality
\begin{equation}
\label{eq42}
\pi<\varphi_W+T<\frac{3\pi}2,
\end{equation}
choose such a small $\varepsilon$ that
\begin{equation}
\label{eq43}
\pi<\varphi_W+T-\varepsilon<\frac{3\pi}2.
\end{equation}
Choose $t_1$ and $t_2$ such that $|t_2-t_1|=\varepsilon$.
Then for such $\varepsilon$ the first inequality in~\eqref{eq41} holds by the virtue of~\eqref{eq43}. The right inequality in~\eqref{eq42}, i.e.,
$T<{3\pi}/2-\varphi_W$, is a consequence of the inequality $T<{\pi}/2$, but for $T\ge{\pi}/2$ traps are absent according to the lemma~\ref{l3}.
Hence, for $\varphi_W\in[\pi/2,\pi]$ and $T>{\pi}-\varphi_W$
there exist $t_1$ and $t_2$ such that $D>0$. This proves the lemma.
\end{proof}

\begin{proof}[of the Theorem~\ref{t2}]
Let us make the inverse transformation from the system~\eqref{eq22} to the system~\eqref{eq1}\enskip
$U\to S^\dag US$, $W\to S^\dag WS$, replace time $t\to th$ and take into account that
$g(t)=f(t)-f_0$. In this case $T\to Th$ and $\varphi_W\to \alpha_W h$. Then  we obtain the statement of the Theorem as a corollary of the lemmas~\ref{l2}--\ref{l4}.
\end{proof}

\section{Numerical analysis of the behaviour of the objective functional in the neighbourhood of the special control $f_0$}\label{s3}

The absence of traps for the objective functional $\mathcal J_W$ for the system~\eqref{eq22}
when $T>\pi/2$ may be illustrated by the numerical analysis. Fig.~\ref{fig1}
shows the plots for the objective value   $J_0=\mathcal J_W[0]$ calculated at the special control
$f=0$, and the probability $P$ that a random control $f$ in a sufficiently small neighbourhood of a special control satisfies the
inequality $\mathcal J_W[f]<J_0$, both as  functions of $(\alpha,\varphi_W)$. Here
$\alpha$ is the angle between the vector
$\mathbf v=\operatorname{Tr}(\boldsymbol{\sigma}V)/2$ and the axis $Ox$ and the angle $\varphi_W$
parametrizes the matrix 
\begin{equation}
\label{eq44}
W=\begin{pmatrix}
e^{i\varphi_W}&0
\\
0&e^{-i\varphi_W}
\end{pmatrix}.
\end{equation}
The matrix  $W$ of this form determines the phase shift gate 
$U_{\phi}=e^{-i\phi/2}W$, where $\varphi_W=-\phi/2$. At each point $(\alpha,\varphi_W)$
the probability $P(\alpha,\varphi_W)$ is estimated as the proportion of realizations of the inequality 
$\mathcal J_W[f]<J_0$ for values of the objective functional  $\mathcal J_W$ calculated for  $M=10^3$ randomly chosen controls in the neighbourhood of the special control $f=0$:
\begin{equation}
\label{eq45}
P(\alpha,\varphi_W)=\frac{\#(f\colon\mathcal J_W[f]<J_0)}{M}.
\end{equation}

\begin{figure}[!ht]
\begin{center}
\includegraphics[scale=0.55]{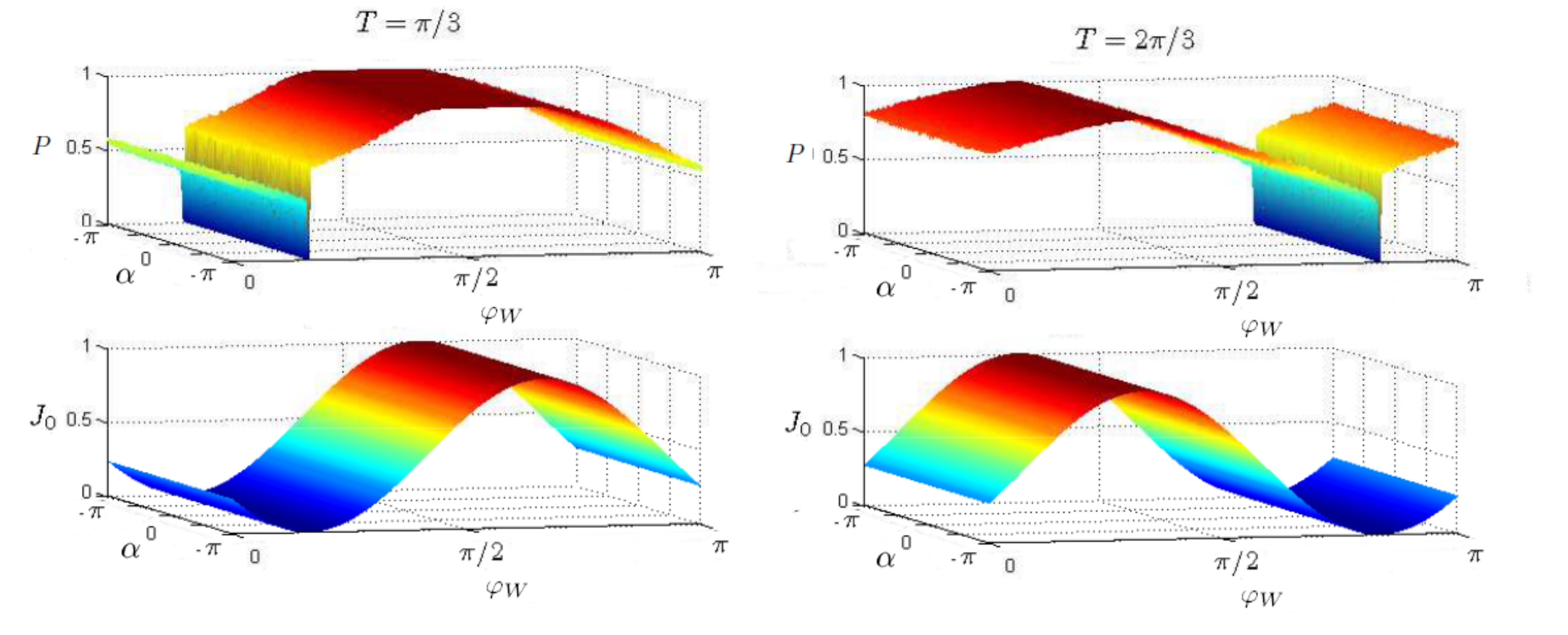}
\vskip3mm
\caption{The plots of the objective value $J_0=\mathcal J_W[0]$ calculated at the special control
$f=0$ and of the probability $P$ that random control $f$ in a sufficiently small neighbourhood of the special control satisfies the
inequality $\mathcal J_W[f]<J_0$. Left: $T=\pi/3$. Right: $T=2\pi/3$.
At each point $(\alpha,\varphi_W)$ the probability $P$ is estimated as the fraction $N_{\mathcal J_W<J_0}$ of  realizations of the inequality $\mathcal J_W<J_0$ among the values of the objective functional $\mathcal J_W$ calculated at $M=10^3$ randomly chosen controls. The random  controls are generated as piecewise constant functions  $f=\sum_{i=1}^{100} a_i\chi_i$, where $\chi_i$ is the characteristic function of the interval $[(i-1)T/100, iT/100]$ and each $a_i$ has normal distribution with unit variance.}
\label{fig1}
\end{center}
\end{figure}

The random controls are generated as piecewise constant functions
$f=\sum_{i=1}^{100}a_i\chi_i$, where $\chi_i$ is the characteristic function of the interval $[(i-\nobreak1)T/100, iT/100]$,
and each $a_i$ has normal distribution with unit variance.
Fig.~\ref{fig1} shows that for $T>\pi/2$ the maxima of the probability $P=P(\alpha,\varphi_W)$
coincide with the points where $J_0(\alpha,\varphi_W)=1$, i.\,e.\ where $f=0$ is a global maximum and, therefore, at this point where are no traps.
In the remaining points $P<1$ and, therefore, there are also no traps. On Fig.~$\ref{fig1}$
the probability $P$  does not depend of the angle $\alpha$, i.\,e.\ on the direction of the vector 
$\mathbf v=\operatorname{Tr}(\boldsymbol{\sigma}V)/2$. This is due to the fact that
values of the objective functional do not depend on this vector, as is stated 
in the following lemma.

\begin{lemma}
\label{l5}
Let $[W,\sigma_z]=0$. Then for any vectors $\mathbf v$ and $\mathbf v'$ determined by the interaction Hamiltonians
$V=(v_x\sigma_x+v_y\sigma_y)$ and
$V'=(v'_x\sigma_x+v'_y\sigma_y)$ in the equation~\eqref{eq22}, the values of the objective functional $\mathcal J_W^\mathbf v[f]$ and $\mathcal J_W^{\mathbf v'}[f]$, calculated at point  $f$ for the $V$ and $V'$, are the same:
$$
\mathcal J_W^\mathbf v[f]=\mathcal J_W^{\mathbf v'}[f].
$$
\end{lemma}

\begin{proof}
Let $U_T^{\mathbf v,f}$ and $U_T^{\mathbf v',f}$ satisfy the equation~\eqref{eq22} with Hamiltonians
$V=v_x\sigma_x+v_y\sigma_y$ and $V'=v'_x\sigma_x+v'_y\sigma_y$. Since
\begin{equation}
\label{eq46}
e^{-i\sigma_z\vartheta/2}(v_x\sigma_x+v_y\sigma_y)e^{i\sigma_z\vartheta/2}
=v'_x\sigma_x+v'_y\sigma_y,
\end{equation}
where $\vartheta=\alpha'-\alpha$ is the angle between the vectors  $\mathbf v'$ and $\mathbf
v$, then the matrix
$$
Z_t=e^{-i\sigma_z{\vartheta}/2}U_t^{\mathbf v,f}e^{i\sigma_z{\vartheta}/2}
$$
satisfies the equation~\eqref{eq22} with the potential $V'=v'_x\sigma_x+v'_y\sigma_y$.
Since $Z_0=\mathbb I$, then
$$
Z_t=e^{-i\sigma_z{\vartheta}/2}U_t^{\mathbf v,f}e^{i\sigma_z{\vartheta}/2}=U_t^{\mathbf
v',f}.
$$
Then
\begin{equation}
\label{eq47}
\mathcal J_W^\mathbf v[f]=\frac14\bigl|\operatorname{Tr}(W^\dag U_t^{\mathbf
v,f})\bigr|^2
=\frac14\bigl|\operatorname{Tr}\bigl(e^{-i\sigma_z\vartheta/2}W^\dag
e^{i\sigma_z\vartheta/2} U_t^{\mathbf v',f}\bigr)\bigr|^2=\mathcal J_W^{\mathbf v'}[f],
\end{equation}
because in the considered case $[W,\sigma_z]=0$ and, therefore, 
$$
e^{-i\sigma_z\vartheta/2}W^\dag e^{i\sigma_z\vartheta/2}=W^\dag.
$$
This proves the lemma.
\end{proof}

\begin{figure}[!ht]
\begin{center}
\includegraphics{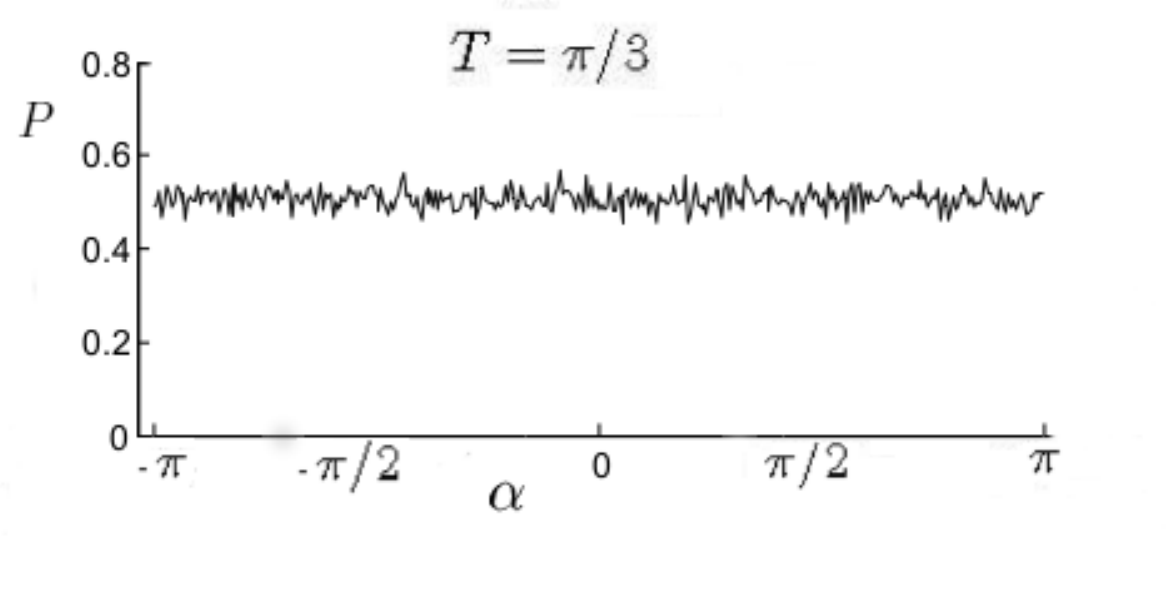}
\vskip3mm
\caption{The probability $P$ that $\mathcal J_{\mathbb H}<J_0$ for $T=\pi/3$.
At each point $\alpha$ the probability $P$ is estimated as the fraction
$N_{\mathcal J_{\mathbb H}<J_0}$
of  realizations of the inequality $\mathcal J_{\mathbb H}<J_0$ among the values of the objective functional
$\mathcal J_{\mathbb H}$ calculated at $M=10^3$ random controls. The random  controls are generated as piecewise constant functions  $f=\sum_{i=1}^{100} a_i\chi_i$, where $\chi_i$ is the characteristic function of the interval
$[(i-1)T/100, iT/100]$ and each $a_i$ has normal distribution with unit variance.}
\label{fig2}
\end{center}
\end{figure}

Fig.~\ref{fig2} shows the probability $P=P(\alpha)$ that
$J_0>\mathcal J_{\mathbb H}[f]$, where the matrix
\begin{equation}
\label{eq48}
\mathbb H=\frac1{\sqrt{2}}
\begin{pmatrix}
1&1
\\
1&-1
\end{pmatrix}
\end{equation}
describes Hadamard gate. From Fig.~\ref{fig2} it is clear that $P(\alpha)\approx 1/2$.
Therefore, where are no traps for any $\alpha$. It is in accordance with Theorem~\ref{t2},
because $[{\mathbb H},\sigma_z]\ne 0$.

\section*{Conclusion} 
In this paper  we prove the Theorem~\ref{t2} about the absence of traps for maximizing the objective functional 
$$
\mathcal J_W[f]=\frac14|\operatorname{Tr}(W^\dag U_T)|^2
$$
for a qubit for small $T$ for almost all $W$. If $[H_0+f_0V,W]\ne0$, then for any $T>0$ all  maxima of the objective functional $\mathcal J_W$ are global.
If $[H_0+f_0V,W]=0$, then the matrix $W$ has the form~\eqref{eq8}. In this case if
\begin{gather*}
\alpha_W\in\biggl(0,\frac{\pi}{2d}\biggr),
\\
d=\biggl\|H_0+f_0V-\frac12\mathbb I\operatorname{Tr} H_0-\frac12f_0\mathbb I
\operatorname{Tr} V\biggr\|,
\end{gather*}
then all  maxima of the objective functional $\mathcal J_W$ are global for any $T>0$.
If $\alpha_W\in[\pi/(2d),\pi/d]$ then all  maxima of the objective functional
$\mathcal J_W$ are global for any $T>\pi/d-\alpha_W$.

\section*{Acknowledgements} The authors thank Corresponding Member of RAS I.\,V.~Volovich for  discussion and useful comments.

\end{document}